\documentclass[12pt,reqno,a4paper]{amsart}
\usepackage{amsmath,amsfonts,amsthm,amssymb,,mathabx,dsfont}
\usepackage{amsxtra,mathrsfs}
\usepackage[utf8]{inputenc}

\newtheorem{theorem}{Theorem}[section]
\newtheorem{proposition}{Proposition}[section]
\newtheorem{lemma}{Lemma}[section]

\theoremstyle{definition}

\theoremstyle{remark}

\numberwithin{equation}{section}

\newcommand{\C}{\mathbb{C}}

\renewcommand{\epsilon}{\varepsilon}

\newcommand{\F}{\mathcal{F}}
\newcommand{\cP}{\mathcal{P}}

\newcommand{\N}{\mathbb{N}}

\newcommand{\R}{\mathbb{R}}

\newcommand{\Z}{\mathbb{Z}}

\newcommand\id{\mathds{1}}

\DeclareMathOperator{\Tr}{Tr}
\DeclareMathOperator{\tr}{Tr}

\newcommand{\spinv}{{\vec S}}
\newcommand{\spin}{S}
\newcommand{\HH}{\mathscr{H}}
\newcommand{\OO}{\mathcal{O}}

\begin{document}

\title[]
{Free energy asymptotics of the quantum Heisenberg spin chain}

\author{Marcin Napi\'orkowski}
\address{Department of Mathematical Methods in Physics, Faculty of Physics, University of Warsaw, Pasteura 5, 02-093 Warsaw, Poland}
\email{marcin.napiorkowski@fuw.edu.pl}

\author{Robert Seiringer}
\address{Institute of Science and Technology Austria, Am Campus 1, 3400 Klosterneuburg, Austria}
\email{robert.seiringer@ist.ac.at}

\thanks{\copyright\, 2021 by the authors. This paper may be  
reproduced, in
its entirety, for non-commercial purposes.}

\begin{abstract} 
We consider the ferromagnetic quantum  Heisenberg model in one dimension, for any spin $S\geq 1/2$. We give upper and lower bounds on the free energy, proving that at low temperature it  is asymptotically equal to the one of an ideal Bose gas of magnons, as predicted by the spin-wave approximation. The trial state used in  the upper bound  yields  an analogous estimate also in the case of two spatial dimensions, which is believed to be sharp at low temperature.
\end{abstract}

\date{March 6, 2021}

\maketitle

%%%%%%%%%%%%%%%%%%%%%%%%%%%%%%%%%%%%%%%%%%%%%%%%%%%%%%%%

\section{Introduction}
The ferromagnetic quantum  Heisenberg model is one of the most important and widely studied models of statistical mechanics. In dimensions $d \geq 3$, the model is widely believed to display long range order at low temperature, but a rigorous proof remains elusive. Based on the concept of long range order, the low temperature properties of the model are usually examined using spin-wave theory. In the  spin-wave approximation one assumes that the low-energy behavior of the system can be described in terms of collective excitations of spins called spin waves. From an equivalent point of view, which dates back to Holstein and Primakoff \cite{HP}, these spin waves are known as bosonic quasiparticles called magnons.  

 The spin-wave approximation has been very successful, predicting for example a phase transition in three and more dimensions, or the $T^{3/2}$ Bloch  magnetization law \cite{B1,B2}. In his seminal 1956 paper \cite{D}, Dyson derived further properties of the  quantum Heisenberg model which, among other things, included the low temperature expansion of the magnetization. 
 
While there was little doubt about the validity of  spin-wave theory in three (or more) dimensions, a rigorous proof of some of its predictions has only recently been given  
in \cite{CGS} (see also  \cite{CorGiuSei-14}). There it was proved that the free energy of the three-dimensional ferromagnetic  quantum  Heisenberg model is, to leading order, indeed given by the expression derived using spin-wave approximation, for any spin $S\geq 1/2$. (See also \cite{CS2,T} for earlier non-sharp upper bounds, or \cite{CG,Benedikter-2017} for results in the large $S$ limit). 

The situation is different in lower dimensions. It has been known since the seminal work of Mermin and Wagner \cite{MW} that the $d=1$ and $d=2$ dimensional quantum Heisenberg models do not exhibit long range order at any non-zero temperature. The low temperature behavior of the system in low dimensions is thus very different from the one in three or higher dimensions, and it is less clear whether spin-wave theory should also be  valid in lower dimensions. 

In 1971 Takahashi \cite{Takahashi-1971} derived a free energy expansion for $d=1$ in the  case $S=1/2$. In this special case the quantum Heisenberg model is exactly solvable via the Bethe ansatz \cite{Bethe-1931}. The spectrum of the (finite size) model can be obtained by solving the corresponding Bethe equations. Under certain assumptions (known as string hypothesis) on the solutions of these equations 
he derived what are now known as thermodynamic Bethe equations, an analysis of  which leads to a formula for the free energy. Later, in \cite{Takahashi-1986} he derived an alternative free energy expansion using (a modified) spin-wave theory (for any $S$, and also in two dimensions). Interestingly, the second terms in the (low temperature) free energy expansions in \cite{Takahashi-1971,Takahashi-1986} do not agree with the predictions of conventional spin-wave theory \cite{B1,B2,D,HP}. (The leading terms do agree, however.)

The thermodynamic Bethe equations have been used not only for the Heisenberg spin chain, but also in other models including the Kondo model \cite{Andrei-80,AndDes-84,AndJer-95,Schlottmann-2001}  or the Gross--Neveu model in high energy physics \cite{AndLow-79}. For more applications of the string hypothesis and its relation to  numerous other models in physics we refer to the review articles \cite{vanTongeren-2016,Maslyuk-2016}.

In the present paper, using different methods, we prove that, to leading order, the formula derived by Takahashi based on the Bethe ansatz and the string hypothesis in \cite{Takahashi-1971} is indeed correct. Our analysis does not use the Bethe ansatz and our result  holds for any spin $S$. It therefore also partly justifies the spin-wave approximation derived in \cite{Takahashi-1986}. We shall utilize some of the methods developed for the three-dimensional case in \cite{CGS}, but novel ingredients are needed to treat the case of lower dimensions, both for the upper and the lower bounds.

%%%%%%%%%%%%%%%%%%%%%%%%%%%%%%%%%%%%%%%%%%%

\section{Model and Main Result}\label{sec:model}

We consider the one-dimensional ferromagnetic quantum Heisenberg model with nearest neighbor interactions. For a chain of length $L$, it is defined in terms of the Hamiltonian
\begin{equation}\label{heisenberg ham 1}
H_L = \sum_{x=1}^{L-1} \left( S^2 - \vec S_x \cdot \vec S_{x+1} \right).
\end{equation}
Here $ \spinv=(S^1,S^2,S^3)$ denote the three components of the spin operators corresponding to spin $S$, i.e., they are the generators of the rotations in a $2S+1$ dimensional representation of $SU(2)$. The Hamiltonian $H_L$ acts on the Hilbert space $\HH_L = \bigotimes_{x=1}^L \C^{2S+1}$. We added a constant $S^2$ for every bond in order to normalize the ground state energy of $H_L$ to zero.

Our main object of study  is the specific free energy
\begin{equation*} 
f_L(\beta,S) = - \frac{1}{\beta L} \ln \left( \Tr e^{-\beta H_L} \right)
\end{equation*}
for $\beta>0$, 
and its thermodynamic limit
\begin{equation} \label{eq:free_energy_1d_thermodynamic}
f(\beta,S) = \lim_{L\to \infty} f_{L}(\beta,S).
\end{equation}
We are interested in the behavior of  $f(S,\beta)$ in the low temperature limit $ \beta \to \infty $ for fixed $S$. Our main result is as follows.

\begin{theorem}\label{thm:main_thm}
Consider the Hamiltonian \eqref{heisenberg ham 1} and the corresponding free energy \eqref{eq:free_energy_1d_thermodynamic}. For any $S\geq 1/2$,
\begin{equation}\label{eq:mainthmd1}
\lim_{\beta\to \infty} f(\beta,S)S^{\frac12} \beta^{\frac32}=C_1:=\frac{1}{2\pi}\int_{\mathbb{R}}\ln \big(1-e^{-p^2}\big)dp=\frac{-\zeta(\frac32)}{2\sqrt{\pi}},
\end{equation}
where $\zeta$ denotes the Riemann zeta function.
\end{theorem}

The proof of Theorem \ref{thm:main_thm} will be given in Sections \ref{sec:up} and \ref{sec:lower}, where we derive quantitative upper and lower bounds, respectively.  The  trial state employed in the derivation of the upper bound can also be used  in $d=2$ dimensions. We refer to Proposition \ref{thm:2D_upper} in Appendix \ref{sec:appendix} for a precise statement and  its proof. A corresponding lower bound for $d=2$ is still missing, however.

The analogue of Theorem~\ref{thm:main_thm} for $d=3$ was proved in \cite{CGS}. While the new tools developed here for the lower bound use the one-dimensional nature of the model in an essential way, they are robust enough to allow for an extension of our results to quasi-one-dimensional systems, like Heisenberg models defined on ladder graphs. Such an extension is rather straightforward and we shall not give the details here.

\section{Boson Representation}\label{sec:bosons}

It is well known that the Heisenberg Hamiltonian can  be rewritten in terms of bosonic creation and annihilation operators \cite{HP}. For any $ x \in [1,\ldots, L] \subset \Z $ we set
\begin{equation}\label{ax upx}
	\spin_{x}^+ =  \sqrt{2S}\, a^\dagger_x \left[ 1 - \frac{a^\dagger_x a_x}{2S} \right]_+^{1/2} \ ,		\quad	\spin_{x}^{-} =  \sqrt{2S}\left[ 1 - \frac{a^\dagger_x a_x}{2S} \right]_+^{1/2} a_x\ ,	\quad	\spin_{x}^3 =  a^\dagger_x a_x - S\,,
\end{equation}
where $a^{\dagger}_{x}, a_x$ are bosonic creation and annihilation operators,  $S^\pm_x = S^1_x \pm i S^2_x$, and $[\, \cdot\,]_+ = \max\{0, \, \cdot\, \}$ denotes the positive part.  The operators $a^\dagger$ and $a$ act on    $f \in \ell^2(\N_0)$ via $(a\, f)(n) = \sqrt{n+1} f(n+1)$ and $(a^\dagger f)(n) = \sqrt{n} f(n-1)$, and satisfy the canonical commutation relations $[a,a^\dagger] = 1$. One readily checks that (\ref{ax upx}) defines a representation of $SU(2)$ of spin $S$, and the operators $\spinv_x$ leave the space $\bigotimes_{x=1}^L \ell^2 ( [0,2S]) \cong \HH_L = \bigotimes_{x=1}^L \C^{2S+1}$, which can  naturally be identified with a subspace of the Fock space $\F_L:=\bigotimes_{x=1}^L \ell^2(\N_0)$, invariant. 

The Hamiltonian $H_L$ in \eqref{heisenberg ham 1}  can be expressed in terms of the bosonic  creation and annihilation operators as
\begin{align} \nonumber
H_L  =   S  \sum_{x=1}^{L-1} \biggl(  & - a^\dagger_x \sqrt{ 1 - \frac {n_x}{2S} }  \sqrt{ 1-\frac{n_{x+1}}{2S}}a_{x+1}  - a^\dagger_{x+1} \sqrt{ 1-\frac{n_{x+1}}{2S}}   \sqrt{ 1 - \frac {n_x}{2S} }  a_x \\ & + n_x + n_{x+1} - \frac 1{S} n_x n_{x+1}  \biggl) \,, \label{hamb}
\end{align}
where we denote the number of particles at site $x$ by $n_x= a^\dagger_x a_x$. It describes a system of bosons hopping on the chain $[1,\ldots L]$ with nearest neighbor {\em attractive} interactions and a hard-core condition preventing more than $2S$ particles to occupy the same site. Also the hopping amplitude depends on the number of particles on neighboring sites, via the square root factors in the first line in (\ref{hamb}). 

In the bosonic representation (\ref{hamb}), the Fock space vacuum $|\Omega\rangle$ (defined by $a_x |\Omega\rangle = 0$ for all $x$) is a ground state of the Hamiltonian $H_L$, and the excitations of the model can be described as bosonic particles in the same way as phonons in crystals. There exists a zero-energy ground state for any particle number less or equal to $2SL$, in fact. While this may not be immediately apparent from the representation (\ref{hamb}), it is a result of the $SU(2)$ symmetry of the model. The total spin is maximal in the ground state, which is therefore $(2SL+1)$-fold degenerate, corresponding to the different values of the $3$-component of the total spin. The latter, in turn, corresponds to the total particle number (minus $SL$) in the bosonic language.

Before we present the proof of Theorem~\ref{thm:main_thm}, we shall briefly explain the additional difficulties compared to the $d=3$ case, and the reason why the proof in \cite{CGS} does not extend to $d=1$. Spin-wave theory predicts that at low temperatures the interaction between spin waves can be neglected to leading order. This means that \eqref{hamb} can  effectively  be replaced by the Hamiltonian of free bosons hopping on the lattice. At low temperature and long wave lengths $\ell\gg 1$, one can work in a continuum approximation where  the last term $-\sum_x n_x n_{x+1}$ in \eqref{hamb} scales as $\ell^{-d}$, while the kinetic energy  scales as $\ell^{-2}$. The interaction terms can thus be expected to be negligible only for $d\geq 3$, and this is indeed what was proved in \cite{CGS}. 
This argument is in fact misleading, as the  attractive interaction term  turns out to be compensated by the correction terms in the kinetic energy coming from the square root factors. Making use of this cancellation will be crucial for our analysis (while it was not needed in \cite{CGS} to derive the free energy asymptotics for $d\geq 3$). 

We note that for $d=1$ and $d=2$ the interaction is strong enough to create bound states between magnons \cite{bs2,bs2b,bs1,bs3,GS}. These occur only at non-zero total momentum, however,  with binding energy much smaller than the center-of-mass kinetic energy at low energies. Hence they 
do not influence the thermodynamic  properties of the system at  low temperature to leading order.

%%%%%%%%%%%%%%%%%%%%%%%%%%%%%%%%%%%%%%%%%%%%%%%%%%%%%%%%%%

\section{Upper Bound}\label{sec:up}

Recall the definition of $C_1$  in \eqref{eq:mainthmd1}. In this section we will prove the following. 

\begin{proposition}\label{ub: pro}
As $ \beta S \to \infty $,  we have 
\begin{equation}\label{fe ub asympt1}
f(\beta,S) \leq C_1 S^{-\frac12} \beta^{-\frac32} \left(1 - \OO( (\beta S)^{-\frac{1}{8}} (\ln \beta S)^{3/4}) \right).\, 
\end{equation}
\end{proposition}

The general structure of the proof will be similar to the corresponding upper bound given in \cite{CGS}. The difference lies in the choice of the trial state, which in contrast to \cite{CGS} allows for  more than one particle on a single site. This is essential in order to capture the desired cancellations explained in the previous section.

\medskip

\noindent {\it Step 1. Localization in Dirichlet boxes.} Our proof will rely on the Gibbs variational principle, which states that 
\begin{equation}\label{varpr}
f_L (\beta,S) \leq  \frac{1}{L} \tr H_L \Gamma  +  \frac{1}{\beta L} \tr  \Gamma \ln \Gamma
\end{equation}
for any positive $\Gamma$ with $\tr \Gamma =1$. We shall confine the particles into smaller intervals, introducing Dirichlet boundary conditions. To be precise, let 
\begin{equation*}
H_{L}^{\rm D} = H_{L}  + 2 S^2 + S(S_1^3 + S_L^3)  
\end{equation*}
be the Heisenberg Hamiltonian on $\Lambda_L := [1,\ldots,L] \subset \Z$ with $S^3_x=-S$ boundary conditions.  Note that $H_{L}^{\rm D} \geq H_{L}$. It is well-known that the thermodynamic limit in \eqref{eq:free_energy_1d_thermodynamic} exists, hence we can assume without loss of generality that $L = k(\ell+1)+1$ for some integers $k$ and $\ell$. By letting all spins  on the boundary of the smaller intervals of side length $\ell$ point maximally in the negative $3$-direction, we obtain the upper bound
\begin{equation*}
f_{L}(\beta,S) \leq \left( 1 + \ell^{-1} \right)^{-1} f_{\ell}^{\rm D}(\beta,S) \ , \quad f_\ell^{\rm D}(\beta,S) := - \frac 1{\beta \ell} \ln \left( \Tr e^{-\beta H_{\ell}^{\rm D}} \right) \,.
\end{equation*}
In particular, by letting $k\to \infty$ for fixed $\ell$, we have 
\begin{equation}\label{eq:localization_upperbound}
f(\beta,S) \leq \left( 1 + \ell^{-1} \right)^{-1} f_{\ell}^{\rm D}(\beta,S) 
\end{equation}
in the thermodynamic limit.

\medskip

\noindent {\it Step 2. Choice of trial state.} To obtain an upper bound on $f_{\ell}^{\rm D}$, we can use the variational principle (\ref{varpr}), with
\begin{equation}\label{def:gamma}
\Gamma = \frac{ \mathcal{P} e^{-\beta  K} \mathcal{P} }{\Tr_\F \mathcal{P} e^{-\beta K }\cP} \,
\end{equation}
where we denote the Fock space $\F\equiv \F_\ell$ for simplicity. Here,  $\mathcal{P}$ is an operator satisfying $0\leq \mathcal{P}\leq 1$, and is defined by
\begin{equation}
\cP=\prod_{x=1}^{\ell}f(n_x) \label{def:cP}
\end{equation}
where
\begin{equation}\label{def:f}
f(n)= \begin{cases} 1 & \text{if} \quad n=0; \\
  \left[ \prod_{j=1}^{n}\left(1-\frac{j-1}{2S}\right) \right]^{\frac12} & \text{if} \quad n=1,2,\ldots, 2S;\\
  0 & \text{if} \quad n>2S.
\end{cases}
\end{equation}
Note that $0\leq \mathcal{P}\leq 1$, and $\mathcal{P}$ is zero if more than $2S$ particles occupy some site. 
The operator $K$ is the Hamiltonian on Fock space $\F$ describing free bosons on $\Lambda_{\ell}=[1,\ldots,\ell]$ with Dirichlet boundary conditions, i.e.,
\begin{align}\nonumber 
 K &=  S \sum_{x,y \in \Lambda_\ell} \left( -\Delta^{\rm D}\right)(x,y) a^\dagger_x a_y \\
 &=  S  \sum_{\langle x,y\rangle\subset \Lambda_\ell} \left( - a^\dagger_x a_y  - a^\dagger_y  a_x  + n_x + n_y  \right) + S (n_1 + n_\ell) 
 \label{hamd}
\end{align}
where $\Delta^{\rm D}$ denotes the Dirichlet Laplacian on $\Lambda_\ell$ and $\langle x,y\rangle$ means that $x$ and $y$ are nearest neighbors. The eigenvalues of $-\Delta^{\rm D}$ are given by
\begin{equation}\label{epsd}
\left\{ \epsilon(p) =  2 (1-\cos(p)) \, : \, p \in \Lambda_\ell^{*\rm D}:= \left\{ \frac {k\pi}{\ell+1}: k\in \{ 1,\dots, \ell\} \right\} \right\}
\end{equation}
with corresponding eigenfunctions $\phi_p(x) = [2/(\ell+1)]^{\frac{1}{2}} \sin(x p)$.

\medskip

\noindent {\it Step 3. Energy estimate.} We shall now give a bound on the energy of the trial state.
\begin{lemma} \label{lem:thp}
On the Fock space $\F=\bigotimes_{x\in\Lambda_\ell} \ell^2(\N_0)$, 
\begin{equation}\label{THP}
\mathcal{P} H_{\ell}^{\rm D} \mathcal{P} \leq K\,.
\end{equation}
\end{lemma}

\begin{proof}
Definition \eqref{def:cP} implies that 
\begin{eqnarray} \label{cPcommutations}
\begin{aligned}
\cP a^\dagger_x &=& \prod_{z\in \Lambda_\ell}f(n_z) a^\dagger_x=a^\dagger_x f(n_x+1)\prod_{\underset{z\neq x}{z\in \Lambda_\ell}}f(n_z)=a_x^\dagger \cP \sqrt{1-\frac{n_x}{2S}}. 
\end{aligned}
\end{eqnarray}
It follows that
\begin{equation} \label{eq:propcP}
\cP  a^\dagger_x \sqrt{ 1 - \frac {n_x}{2S} }  \sqrt{ 1-\frac{n_y}{2S}}a_y \cP = a^\dagger_x \cP^2 \big(1-\frac{n_x}{2S}\big)\big(1-\frac{n_y}{2S}\big) a_y.
\end{equation}
With the aid of \eqref{cPcommutations} and \eqref{eq:propcP} one  checks that
\begin{eqnarray*}
\begin{aligned}
\mathcal{P} H_{\ell}^{\rm D}  \mathcal{P} &= S\sum_{\langle x,y\rangle\subset \Lambda_\ell} (a^\dagger_x - a^\dagger_y) \cP^2 \big(1-\frac{n_x}{2S}\big)\big(1-\frac{n_y}{2S}\big)(a_x-a_y) \\ & \quad + S \sum_{x\in\{1,\ell\}} a^\dagger_x \cP^2 \big(1-\frac{n_x}{2S}\big) a_x.
\end{aligned}
\end{eqnarray*}
The desired bound  \eqref{THP} then follows directly from $\cP^2 \big(1-\frac{n_x}{2S}\big)\big(1-\frac{n_y}{2S}\big)\leq 1$ and $\cP^2 \big(1-\frac{n_x}{2S}\big)\leq 1$.
\end{proof}

We conclude that
\begin{equation} \label{energybound}
\tr H_{\ell}^{\rm D} \Gamma \leq \frac{\tr_\F K e^{-\beta K}}{\tr_\F \cP e^{-\beta K}\cP}\, .
\end{equation}
As a next step, we will show that $\tr_\F \mathcal{P} e^{-\beta K}\cP$ is close to $\tr_\F e^{-\beta K}$ for  $\ell \ll (\beta S)^\frac23 $. The following lemma is an adaptation of the corresponding result in \cite[Lemma~4.3]{CGS}.

\begin{lemma}\label{lem:ppex}
We have
\begin{equation}\label{ppex}
 \frac{ \Tr_\F \mathcal{P} e^{-\beta K}\cP}{\Tr_\F e^{-\beta K} } \geq 1- \left( \frac {\pi^2}{12}\right)^2  \frac{\ell (\ell+1)^2}{(\beta S)^2}. 
 \end{equation} 
\end{lemma}

\begin{proof}
Using that $f(n_x)\leq 1$ and that $f(n_x)=1$ if $n_x\in \{0,1\}$, we have
\begin{equation}\label{omp}
1-\mathcal{P}^2  \leq \sum_{x=1}^\ell (1-f^2(n_x)) \leq \frac 12\sum_{x=1}^\ell  n_x (n_x-1) = \frac 12 \sum_{x=1}^\ell a^\dagger_x a^\dagger_x a_x a_x \,.
\end{equation}
Wick's rule for Gaussian states therefore implies that 
\begin{equation}\label{eq:gaussian_bound_energy}
\frac{ \Tr_\F \mathcal{P} e^{-\beta K}\cP}{\Tr_\F e^{-\beta K} } \geq 1- \frac 12 \sum_{x=1}^\ell \frac{ \Tr_\F a^\dagger_x a^\dagger_x a_x a_x e^{-\beta K}}{\Tr_\F e^{-\beta K} } = 1 -\sum_{x=1}^\ell \left( \frac{ \Tr_\F n_x e^{-\beta K}}{\Tr_\F e^{-\beta K} }\right)^2 \,.
\end{equation} 
Moreover,
\begin{equation*}
 \frac{ \Tr_\F n_x e^{-\beta K}}{\Tr_\F e^{-\beta K}} = \frac 1{ e^{\beta S (-\Delta^{\rm D})} -1 }(x,x) =\sum_{p \in \Lambda_\ell^{*\rm D}}\frac{|\phi_p(x)|^2}{e^{\beta S \epsilon(p)}-1} \leq \frac{2}{\ell +1}\sum_{p \in \Lambda_\ell^{*\rm D}}\frac{1}{e^{\beta S \epsilon(p)}-1}. 
 \end{equation*}
By using $(e^{x}-1)^{-1}\leq x^{-1}$ for $x\geq 0$ in the last sum, as well as $1-\cos x\geq \frac{2x^2}{\pi^2}$ for $x\in(0,\pi)$, this gives 
\begin{equation}\label{eq416}
 \frac{ \Tr_\F n_x e^{-\beta K}}{\Tr_\F e^{-\beta K}} 
 \leq \frac{\ell +1}{2 \beta S } \sum_{n=1}^\ell  \frac 1{n^2}  
  \leq  \frac{\pi^2}{12} \frac{\ell +1}{\beta S } \,.
\end{equation}
Inserting this bound into \eqref{eq:gaussian_bound_energy} yields the desired result.
\end{proof}

\smallskip

\noindent {\it Step 4. Entropy estimate.} It remains to give a lower bound on $-\tr \Gamma\ln \Gamma$,  the entropy of $\Gamma$. We proceed in the same way as in \cite[Lemma~4.4]{CGS}.

\begin{lemma}\label{lem:ent}
We have
\begin{align*} \label{entb}
\frac 1 \beta \Tr \Gamma \ln \Gamma &  \leq - \frac 1 \beta \ln \left(  \tr_\F \mathcal{P} e^{-\beta K}\cP \right) -  \frac{\Tr_\F K e^{-\beta K} }{\Tr_\F \mathcal{P} e^{-\beta K}\cP }  \\ & \quad +  S  \left( \frac{\pi^2}{12} \right)^2  \frac{ \ell(\ell+1)^3}{(\beta S)^{7/2}} \left[ \frac {\sqrt\pi \zeta(3/2)}{8} + \frac { (\beta S)^{1/2}}{\ell} \right] \frac{\Tr_\F e^{-\beta K} }{\Tr_\F \mathcal{P} e^{-\beta K}\cP }    \,.  
\end{align*}
\end{lemma}

\begin{proof}
We have 
\begin{equation*}
\Tr \Gamma \ln \Gamma = - \ln \left( \tr_\F \mathcal{P} e^{-\beta K}\cP\right) + \frac 1{\Tr_\F \mathcal{P} e^{-\beta K} \cP} \Tr_\F \mathcal{P} e^{-\beta K} \mathcal{P} \ln \left( \mathcal{P} e^{-\beta K} \mathcal{P} \right)\,.
\end{equation*}
Using the operator monotonicity of the logarithm, as well as the fact that the spectra of $\mathcal{P} e^{-\beta K} \mathcal{P}$ and $e^{-\beta K/2} \mathcal{P}^2 e^{-\beta K/2}$ agree, we can bound
\begin{align*}\nonumber
 \Tr_\F & \mathcal{P} e^{-\beta K} \mathcal{P} \ln \left(  \mathcal{P} e^{-\beta K} \mathcal{P}  \right) =  \Tr_\F e^{-\beta K/2} \mathcal{P}^2 e^{-\beta K/2} \ln \left(  e^{-\beta K/2} \mathcal{P}^2 e^{-\beta K/2} \right) \\
  & \leq  \Tr_\F  e^{-\beta K/2} \mathcal{P}^2 e^{-\beta K/2} \ln e^{-\beta K}  = - \beta  \Tr_\F K \cP^2 e^{-\beta K} \,.
\end{align*}
Hence
\begin{equation}
\Tr \Gamma \ln \Gamma \leq - \ln \left(  \tr_\F \mathcal{P} e^{-\beta K} \cP \right) - \beta \frac{\Tr_\F K e^{-\beta K} }{\Tr_\F \mathcal{P} e^{-\beta K} \cP } + \beta  \frac{\Tr_\F K (1-\mathcal{P}^2) e^{-\beta K} }{\Tr_\F \mathcal{P} e^{-\beta K}\cP } \,.
\end{equation}
In the last term, we can bound  $1-\mathcal{P}^2$ as in (\ref{omp}), and evaluate the resulting expression using Wick's rule. With $\phi_p$ the eigenfunctions of the Dirichlet Laplacian, displayed below Eq.~(\ref{epsd}), we obtain
\begin{equation}
\begin{aligned}\label{eq:entestimate1D}
\frac{ \Tr_\F K n_x(n_x-1) e^{-\beta K} }{\Tr_\F e^{-\beta K} } & =  \left(\frac{ \Tr_\F n_x e^{-\beta K} }{\Tr_\F e^{-\beta K} } \right)^2  \sum_{p \in \Lambda_\ell^{*\rm D}}  \frac {2S \epsilon(p) }{e^{\beta  S \epsilon(p)} -1}   \\  & \quad + \frac{ \Tr_\F  n_x e^{-\beta K} }{\Tr_\F e^{-\beta K} }    \sum_{p \in \Lambda_\ell^{*\rm D}}  \frac { S\epsilon(p)  |\phi_p(x)|^2 }{ \left( \sinh \tfrac 12 \beta S \epsilon(p) \right)^2 } \,.
\end{aligned}
\end{equation}
The expectation value of $n_x$ can be bounded independently of $x$ as in \eqref{eq416}. When summing over $x$, we can use the normalization $\sum_x |\phi_p(x)|^2 =1$. 
To estimate the sums over $p$ we proceed similarly as in the proof of Lemma~\ref{lem:ppex} to obtain
\begin{align*}
\sum_{p \in \Lambda_\ell^{*\rm D}}  \frac {2S \epsilon(p) }{e^{\beta  S \epsilon(p)} -1} & \leq \frac{\ell+1}{\pi} \int_{0}^\pi   \frac {2S \epsilon(p) }{e^{\beta  S \epsilon(p)} -1} dp \leq   \frac{\ell+1}{\pi^3} \int_{0}^\pi   \frac {8 S p^2}{e^{4\beta  S  p^2/\pi^2 } -1} dp \\
& \leq  S \frac{\ell+1}{(\beta S)^{3/2}}  \int_0^\infty \frac{p^2}{e^{p^2} - 1} dp = S \frac{\ell+1}{(\beta S)^{3/2}} \frac{\sqrt\pi}4 \zeta(3/2)
\end{align*}
and 
\begin{align*}
 \sum_{p \in \Lambda_\ell^{*\rm D}}  \frac { S\epsilon(p)  }{ \left( \sinh \tfrac 12 \beta S \epsilon(p) \right)^2 }  \leq  \frac 4{S\beta^2 }  \sum_{p \in \Lambda_\ell^{*\rm D}}  \frac { 1  }{ \epsilon(p) }  \leq   \frac {(\ell+1)^2}{S\beta^2 }  \sum_{n=1}^\ell   \frac { 1  }{ n^2  } \leq  \frac {\pi^2}6 \frac {(\ell+1)^2}{S\beta^2}  \,.
 \end{align*}
In combination this yields the desired bound.
\end{proof}

\smallskip

\noindent {\it Step 5. Final estimate.}  The Gibbs variational principle \eqref{varpr} together with \eqref{energybound}, Lemma \ref{lem:ent}   and Lemma  \ref{lem:ppex}  implies that  
\begin{equation*}
\begin{aligned}
f_\ell^{\rm D}(\beta,S)& \leq - \frac{1}{\beta\ell} \ln \left( \tr_\F \mathcal{P} e^{-\beta K}\cP \right)  +  C S\frac{\ell^3}{(\beta S)^{7/2}}\frac{\Tr_\F e^{-\beta K} }{\Tr_\F \mathcal{P} e^{-\beta K}\cP }\\
&\leq -\frac{1}{\beta \ell}\ln \left( \tr_\F e^{-\beta K} \right) -\frac{1}{\beta \ell}\ln \left(1-\frac{C \ell^3}{(\beta S)^2}\right) + C S\frac{\ell^3}{(\beta S)^{7/2}}
\end{aligned}
\end{equation*}
for a suitable constant $C>0$, as long as $C (\beta S)^{1/2} \leq \ell \ll (\beta S)^{2/3}$. 
The first term on the right side in the second line of the expression above equals
\begin{equation}
 - \frac 1 {\beta \ell} \ln \left( \tr_\F e^{-\beta K}\right) = \frac 1 {\beta\ell} \sum_{p\in \Lambda_\ell^{*\rm D} } \ln ( 1- e^{-\beta S \epsilon(p)} )\,.
\end{equation}
By monotonicity, we can bound the sum by the corresponding integral, 
\begin{equation}\label{up:rie} 
 \frac 1 {\beta\ell} \sum_{p\in \Lambda_\ell^{*\rm D} } \ln ( 1- e^{-\beta S \epsilon(p)} ) \leq  \frac{1}{\pi \beta} \left( 1 +\ell^{-1} \right) \int_{\frac \pi {\ell +1}}^{\pi} \ln ( 1- e^{-\beta S \epsilon(p)} ) dp\,, 
\end{equation}
which is of the desired form, except for the missing part
$$
- \frac{1}{\pi \beta}  \int_0^{\frac \pi {\ell +1}} \ln ( 1- e^{-\beta S \epsilon(p)} ) dp \leq - \frac 1 {\beta(\ell+1)} \int_0^1 \ln \left( 1 - e^{- \frac{4 \beta S}{(\ell+1)^2} p^2} \right) dp = \OO \left( \frac{ \ln (\ell^2/(\beta S)) }{ \beta \ell}\right) 
 $$
 for $\ell \gg (\beta S)^{1/2}$. 
Since $\epsilon(p)\leq p^2$ we further have
\begin{align}\nonumber
\frac{1}{\beta \pi} \int_{0}^\pi \ln ( 1- e^{-\beta S \epsilon(p)} ) dp  & \leq \frac{1}{\pi \beta} \int_{0}^\infty \ln ( 1- e^{-\beta S p^2} ) dp + \frac C {\beta (\beta S)^\alpha }  \\
& =C_1 S^{-1/2} \beta^{-3/2} + \frac C {\beta (\beta S)^\alpha }  \nonumber
\end{align}
for arbitrary $\alpha>0$, some $C>0$ (depending on $\alpha$), and $C_1$ defined in \eqref{eq:mainthmd1}. For $(\beta S)^{2/3} \gg \ell \gg (\beta S)^{1/2}$  all the error terms are small compared to the main term. The desired upper bound stated in Proposition (\ref{ub: pro}) is obtained by combining the estimate above with \eqref{eq:localization_upperbound} and choosing $\ell = C (\beta S)^{5/8} (\ln \beta S)^{1/4}$. \hfill\qed

\section{Lower bound}\label{sec:lower}
Recall the definition \eqref{eq:mainthmd1} of $C_1$.  In this section we shall prove the following. 

\begin{proposition}\label{prop:lower}
As $ \beta S \to \infty $, we have 
\begin{equation*}
f(\beta,S) \geq C_1 S^{-\frac12} \beta^{-\frac32} \left(1 + \OO( (\beta S)^{-\frac{1}{12}}(\ln\beta S)^{1/2} (\ln \beta S^3)^{\frac13}) \right).
\end{equation*}
\end{proposition}

Note that in contrast to the upper bound in Prop.~\ref{ub: pro}, the  lower bound above is not entirely uniform in $S$. Indeed, one has $\ln(\beta S^3)= \ln(\beta S)+\ln S^2$ and hence $S$ is not allowed to grow arbitrarily fast compared to $\beta S$. To obtain a uniform bound, one can combine our results with the method in \cite{CG} where the case $S\to \infty$ for fixed $\beta S$ was analyzed. 

 The remainder of this section is devoted to the proof of Prop.~\ref{prop:lower}. For clarity, the presentation will be divided into several steps. Some of them will use results from \cite{CGS}.
 
 \medskip

\noindent {\it Step 1. Localization.} Recall the definition \eqref{heisenberg ham 1} of the Hamiltonian $H_L$.   For a lower bound, we can drop a term  $( S^2 - \vec S_\ell \cdot \vec S_{\ell+1})$ from the Hamiltonian, which leads to the subadditivity 
\begin{equation} \label{eq:subbadit}
L f_L(\beta, S) \geq \ell f_{\ell}(\beta, S) + (L-\ell) f_{L-\ell}(\beta,S)
\end{equation}
for $1\leq \ell \leq L-1$. By applying this repeatedly, one readily finds that
$$
f(\beta,S) \geq f_\ell(\beta,S)
$$
for any $\ell \geq 1$. We shall choose $\ell$ large compared with the thermal wave length, i.e., $\ell \gg (\beta S)^{1/2}$.

\medskip

\noindent {\it Step 2. Lower bound on the Hamiltonian.} Recall that the total spin operator is defined as $\vec S_{\rm tot} = \sum_{x=1}^\ell \vec S_x$. It follows from the theory of addition of angular momenta that 
\begin{equation}\label{eq:totalspinsquare}
\vec S_{\rm tot}^2 = T(T+1) \ \text{with\ } \sigma(T)=\{0,1,\ldots, S\ell\}\,,
\end{equation}
where $\sigma$ denotes the spectrum. 
We will use the following bound on the Hamiltonian.
\begin{lemma}
With $T$ defined in \eqref{eq:totalspinsquare}, we have 
\begin{equation} \label{eq:lowerboundHam}
H_\ell \geq \frac 2{\ell^3} ( S\ell (S \ell +1)  -  \vec S_{\rm tot}^2 ) \geq \frac {2S}{\ell^2}\left( S\ell - T\right).
\end{equation}
\end{lemma}
\begin{proof}
It was shown in \cite[Eq. (5.6)]{CGS} that
$$
(S^2-\vec S_x\cdot \vec S_y) + (S^2-\vec S_y\cdot \vec S_z)\geq \frac 12 (S^2-\vec S_x\cdot \vec S_z)  
$$
for three distinct sites $x,y,z$, and consequently that
$$
(y-x) \sum_{w=x}^{y-1} \left( S^2 - \vec S_w \cdot \vec S_{w+1} \right) \geq \frac 1{2}  (S^2-\vec S_x\cdot \vec S_{y})  
$$
for any $x<y$. After summing the above bound over all $1\leq x < y \leq \ell$, we obtain
\begin{align*}
\sum_{1\leq x<y\leq \ell}  (S^2-\vec S_x\cdot \vec S_{y})  & \leq 2 \sum_{1\leq x<y\leq \ell} (y-x) \sum_{w=x}^{y-1} \left( S^2 - \vec S_w \cdot \vec S_{w+1} \right)
\\ & = 2 \sum_{w=1}^{\ell-1} \left( S^2 - \vec S_w \cdot \vec S_{w+1} \right)  \sum_{x=1}^w \sum_{y=w+1}^\ell (y-x).
\end{align*}
We have
$$
\sum_{x=1}^w \sum_{y=w+1}^\ell (y-x) = \frac \ell 2 w (\ell - w) \leq \frac{\ell^3}{8}
$$
for $1\leq w \leq \ell-1$, 
and hence
$$
H_\ell \geq \frac 4{\ell^3} \sum_{1\leq x<y\leq \ell}  (S^2-\vec S_x\cdot \vec S_{y}) =\frac 2{\ell^3} ( S\ell (S \ell +1)  -  \vec S_{\rm tot}^2 ). 
$$
As $\vec S_{\rm tot}^2 = T(T+1)$ we thus have 
$$
H_\ell \geq \frac {2S}{\ell^2}\left( S\ell+1 - \frac{T(T+1)}{S\ell}\right).
$$
The final bound \eqref{eq:lowerboundHam} then follows from the fact that $T\leq S\ell$.
\end{proof}

Note that Lemma~\ref{eq:lowerboundHam} implies, in particular, a lower bound of  $2S \ell^{-2}$  on the spectral gap of $H_\ell$ above its ground state energy. For $S=1/2$, it follows from the work in  \cite{CLR} that the exact spectral gap  equals  $ ( 1- \cos(\pi/\ell))$ (which is $\frac 12 \pi^2 \ell^{-2}$ to leading order for large $\ell$).   

\smallskip

\noindent {\it Step 3. Preliminary lower bound on free energy.} With the aid of \eqref{eq:lowerboundHam} we shall now prove the following preliminary lower bound on the free energy.
\begin{lemma}\label{lem53}
Let
\begin{equation}
\ell_0 := \sqrt{\frac{4 \beta S}{\ln \beta S }} \label{def:ell0}
\end{equation}
and assume that $\ell \geq \ell_0/2$. Then, for $\beta  S$ sufficiently large, we have
\begin{equation}\label{eq:preelimfreeenergylower}
f_\ell(\beta,S) \geq -  C\frac { \left( \ln \beta S\right)^{1/2}}{\beta^{3/2} S^{1/2}} \ln \beta S^3
\end{equation}
for some constant $C>0$.
\end{lemma}

\begin{proof}
With the aid of \eqref{eq:lowerboundHam} and the $SU(2)$ symmetry we have
\begin{align*}
\Tr e^{-\beta H_\ell} & \leq   \sum_{n=0}^{\lfloor S\ell\rfloor} e^{   {-2\beta S}{\ell^{-2}}   n   }  \Tr \id_{T=S\ell -n} \\ & =  \sum_{n=0}^{\lfloor S\ell\rfloor} e^{   {-2\beta S}{\ell^{-2}}   n   } \left( 2(S\ell -n)+1\right)  \Tr \id_{T=S\ell -n}\id_{S_{\rm tot}^3 = n-S\ell}
\\ & \leq (2 S\ell +1)  \sum_{n=0}^{\lfloor S\ell\rfloor} e^{   {-2\beta S}{\ell^{-2}}   n   }  \Tr \id_{S_{\rm tot}^3 = n-S\ell}.
\end{align*}
The last trace equals the number of ways $n$ indistinguishable particles can be distributed over $\ell$ sites, with at most $2S$ particles per site. Dropping this latter constraint for an upper bound, we obtain
$$
\Tr e^{-\beta H_\ell}  \leq (2S \ell+1)  \left ( 1 -  e^{   {-2\beta S}{\ell^{-2}}      }\right)^{-\ell}  \,.
$$
In particular,
\begin{equation}\label{fbo}
f_\ell(\beta,S) \geq  -\frac 1{\beta \ell} \ln  (1+ 2 S\ell )  + \frac {1}{\beta } \ln   \left ( 1 -  e^{   {-2\beta S}{\ell^{-2}}      }\right)\,.
\end{equation}
For large $\beta S$, this expression is minimized when $\ell \approx \ell_0$ with $\ell_0$ given in \eqref{def:ell0}. If $\ell_0/2 \leq \ell \leq \ell_0$, we can use the lower bound on $\ell$ in the first term in \eqref{fbo}, and the upper bound on the second, to obtain
\begin{equation}\label{fbo2}
f_\ell(\beta,S) \geq  -\frac { (\ln \beta S)^{1/2}}{\beta ( \beta S)^{1/2}} \ln  \left(1+  2 S (\beta S)^{1/2} (\ln \beta S)^{-1/2} \right)  + \frac {1}{\beta } \ln   \left ( 1 -  (\beta S)^{-1/2}  \right)\,,
\end{equation}
which is of the desired form. 
If $\ell > \ell_0$,   we can divide the interval $[1,\ell]$ into smaller ones of size  between $\ell_0/2$ and $\ell_0$. Using the subadditivity \eqref{eq:subbadit} we conclude \eqref{fbo2} also in that case.
\end{proof}

\smallskip

\noindent {\it Step 4. Restriction to low energies.} 
For any $E>0$, we have
\begin{align*}
\Tr e^{-\beta H_\ell} & \leq \Tr e^{-\beta H_\ell} \id_{H_\ell < E} +  e^{-\beta E/2} \Tr e^{-\beta H_\ell/2 } \id_{H_\ell \geq E}
\\ & \leq \Tr e^{-\beta H_\ell} \id_{H_\ell < E}  + e^{-\beta (E + \ell f_\ell(\beta/2,S))/2}.
\end{align*}
In particular, with the choice
$$
E = E_0(\ell,\beta,S) := - \ell f_\ell(\beta/2,S)
$$
this gives
\begin{equation}\label{tg}
\Tr e^{-\beta H_\ell} \leq 1 + \Tr e^{-\beta H_\ell} \id_{H_\ell < E_0}.
\end{equation}
Using the $SU(2)$ invariance, we can further write
\begin{align}\nonumber
 \Tr e^{-\beta H_\ell} \id_{H_\ell < E_0} & = \sum_{n=0}^{\lfloor S\ell\rfloor} (2(S\ell -n) +1)  \Tr e^{-\beta H_\ell} \id_{H_\ell < E_0} \id_{T=S\ell -n} \id_{S_{\rm tot}^3=n-S\ell}
 \\ & \leq  (2 S\ell  +1) \sum_{n=0}^{\lfloor S\ell\rfloor}   \Tr e^{-\beta H_\ell} P_{E_0,n} \label{eq56}
\end{align}
where
\begin{equation} \label{def:PE0}
P_{E_0,n} =  \id_{H_\ell < E_0} \id_{T=S\ell -n} \id_{S_{\rm tot}^3=n-S\ell}.
\end{equation}
In other words, we can restrict the trace to states with $S_{\rm tot}^3$ being as small as possible (given $\vec S_{\rm tot}^2$). In the particle picture discussed in Section~\ref{sec:bosons}, this amounts to particle number $\mathcal{N} = S\ell - T = n$. Because of \eqref{eq:lowerboundHam}, we have $E_0 > H_\ell \geq 2 S n/\ell^2$  on the range of $P_{E_0,n}$, hence the sum in \eqref{eq56} is restricted to 
\begin{equation} \label{def:N_0}
n < N_0 := \frac {E_0 \ell^2}{2S}.
\end{equation}

\smallskip

\noindent {\it Step 5. A Laplacian lower bound.}  With the aid of the Holstein--Primakoff representation \eqref{ax upx}, we can equivalently write the Hamiltonian $H_\ell$ in terms of bosonic creation and annihilation operators as
\begin{equation}\label{newham}
H_\ell = S\sum_{x=1}^{\ell-1} \left( a^\dagger_{x+1} \sqrt{ 1 - \frac{n_x}{2S} } - a^\dagger_x \sqrt{ 1 - \frac{n_{x+1}}{2S}}  \right) \left( a_{x+1} \sqrt{ 1 - \frac{n_x}{2S} } - a_x \sqrt{ 1 - \frac{n_{x+1}}{2S}}  \right) 
\end{equation}
where $n_x =a^\dagger_x a_x \leq 2S$. Note that written in this form, the Hamiltonian $H_\ell$ is manifestly positive, contrary to \eqref{hamb}. 

 Let $\mathcal{N} = \sum_{x} n_x = \ell S + S_{\rm tot}^3$ denote the total number of bosons. 
States  $\Psi$ with $n$ particles, i.e., $\mathcal{N} \Psi = n\Psi$, are naturally identified with $n$-boson wave functions\footnote{Here $\ell^2_{\rm sym}(A)$ denotes the Hilbert space of square-summable sequences on $A$ invariant under permutations} in $\ell^2_{\rm sym}([1,\ell]^n)$ via 
$$
\Psi = \frac 1{\sqrt{n!}} \sum_{1\leq x_1,\dots,x_n \leq \ell} \Psi(x_1,\dots,x_n)  { a^\dagger_{x_1} \cdots a^\dagger_{x_n} |\Omega\rangle } \,, 
$$
where $|\Omega\rangle$ denotes the vacuum (which corresponds to the state with all spins pointing maximally down). Using \eqref{newham}, we have in this representation 
\begin{align*}\nonumber
\langle \Psi | H_\ell \Psi \rangle & = S n \sum_{x=1}^{\ell-1} \sum_{x_1,\dots,x_{n-1}} \left | \Psi(x+1,x_1,\dots,x_{n-1}) \sqrt{ 1 - \frac{\sum_{k=1}^{n-1} \delta_{x,x_k}}{2S} } \right. \\  & \left. \qquad\qquad\qquad\qquad\quad  -  \Psi(x,x_1,\dots,x_{n-1}) \sqrt{ 1 - \frac{\sum_{k=1}^{n-1} \delta_{x+1,x_k}}{2S} }\right|^2.
\end{align*}
Because of permutation-symmetry, we can also write this as
\begin{align*}\nonumber
\langle \Psi | H_\ell \Psi \rangle & = S \sum_{j=1}^n  \sum_{\substack{x_1,\dots,x_{n} \\ x_j \leq \ell-1}} \left | \Psi(x_1,\dots, x_j+1, \dots x_{n}) \sqrt{ 1 - \frac{\sum_{k, k\neq j} \delta_{x_j,x_k}}{2S} } \right. \\  & \left. \qquad\qquad\qquad\quad  -  \Psi(x_1,\dots,x_j, \dots x_{n}) \sqrt{ 1 - \frac{\sum_{k, k\neq j} \delta_{x_j+1,x_k}}{2S} }\right|^2\,.
\end{align*}

For a lower bound, we can restrict the sum over $x_1,\dots,x_{n}$ to values such that $x_k\neq x_l$ for all $k\neq \ell$. For a given $j$, we can further restrict to $x_k\neq x_j+1$ for all $k\neq j$. In this case, the square root factors above are equal to $1$. In other words, we have the lower bound
$$
\langle \Psi | H_\ell \Psi \rangle \geq  \frac S 2 \sum_{\substack{X,Y \in \mathcal{X}_{\ell,n}  \\ |X-Y| = 1}} \left| \Psi(X) - \Psi(Y)\right|^2
$$
where the sum is over the set $\mathcal{X}_{\ell,n} := \{ [1,\ell]^n : x_i \neq x_j \forall i\neq j\} $, and $|X-Y| = \sum_{i=1}^n |x_i-y_i|$. Note that we have to assume  that $\ell \geq n$ for the set $\mathcal{X}_{\ell,n}$ to be non-empty. The factor $1/2$ arises from the fact that particles are allowed  to hop both left and right, i.e., each pair $(X,Y)$ appears twice in the sum. Note also that the above inequality is actually an equality for $S=1/2$, since in this case no two particles can occupy the same site.

On the set $\{ 1\leq x_1 < x_2 < \dots <x_n\leq \ell\} \subset \mathcal{X}_{\ell,n}$ define the map 
$$
V (x_1, \dots, x_n) = (x_1, x_2-1, x_3-2 , \dots, x_n - n+ 1)
$$
and extend it to the set $\mathcal{X}_{\ell,n} = \{ [1,\ell]^n : x_i \neq x_j \forall i\neq j\} $ via permutations. In other words, $V$ maps $x_i$ to $x_i - k_i$ where $k_i$ denotes the number of $x_j$ with $x_j < x_i$. As a map from $\mathcal{X}_{\ell,n}$ to $[1,\ell-n+1]^n$, $V$ is clearly surjective, but it is not injective. Points in $[1,\ell-n+1]^n$ with at least two coordinates equal have more than one pre-image under $V$. The pre-images are unique up to permutations, however, hence we can define a map $\mathbb{V}: \ell^2_{\rm sym}([1,\ell]^n) \to \ell^2_{\rm sym}([1,\ell-n+1]^n)$ via
\begin{equation}\label{def:V}
\mathbb V \Psi (V(X)) = \Psi(X) \quad \text{for $X\in \mathcal{X}_{\ell,n}$}.
\end{equation}
%
%
%
% bosonic wave function $\Phi$ on $[1,\ell-n+1]^n$ by
%\begin{equation} \label{def:mapV}
%\Phi(V(X)) = \Psi(X) \quad \text{for $X\in \mathcal{X}_{\ell,n}$}.
%\end{equation}
We then have
\begin{align*}
& \sum_{\substack{X,Y \in \mathcal{X}_{\ell,n}  \\ |X-Y| = 1}} \left| \Psi(X) - \Psi(Y)\right|^2 \\ &= \sum_{A,B \in [1,\ell-n+1]^n} \left| \mathbb{V} \Psi(A)- \mathbb{V} \Psi(B) \right|^2 \sum_{X\in V^{-1}(A), Y\in V^{-1}(B)} \chi_{|X-Y|=1}\,.
\end{align*}
For every pair $(A,B)\in [1,\ell-n+1]^n$ with $|A-B|=1$, there exists at least one pair $(X,Y)\in \mathcal{X}_{\ell,n} $ with $|X-Y|=1$ in the pre-image of $V$. In other words, the last sum above is greater or equal to $1$ if $|A-B|=1$. We have thus proved the following statement.

\begin{proposition}\label{prop:hamiltonianLaplacianBound}
Let  $\mathbb{V}: \ell^2_{\rm sym}([1,\ell]^n) \to \ell^2_{\rm sym}([1,\ell-n+1]^n)$  be defined in \eqref{def:V}. Then
$$
\id_{\mathcal{N}=n} H_\ell \geq S \mathbb{V}^\dagger (-\Delta_n^{\ell-n+1}) \mathbb{V},
$$
where  $\Delta_n^\ell$ denotes the Laplacian\footnote{This is the graph Laplacian, with free (or Neumann) boundary conditions.} on $[1,\ell]^n$. 
\end{proposition}

\smallskip

\noindent {\it Step 6. Bounds on the two-particle density.} 
We will use Prop.~\ref{prop:hamiltonianLaplacianBound} and  the min-max principle  to obtain a lower bound on the eigenvalues of $H_\ell$. For this purpose we need an estimate on the norm of $\mathbb{V}\Psi$.

For $\Psi\in \ell^2_{\rm sym}([1,\ell]^n)$ with $\|\Psi\|=1$, we let 
$$\rho_\Psi(x,y) = \langle \Psi | a^\dagger_x a^\dagger_y a_y a_x \Psi\rangle$$ 
denote its  two-particle density. 

\begin{lemma}\label{lem:phinormlower}
Let $\Psi\in \ell^2_{\rm sym}([1,\ell]^n)$ with $\|\Psi\|=1$. %, and let $\rho(x,y) = \langle \Psi | a^\dagger_x a^\dagger_y a_y a_x \Psi\rangle$ be its  two-particle density. 
Then
\begin{equation}\label{eq:normphi_lowerbound}
\|\mathbb{V}\Psi\|^2 \geq 1 - \frac 12 \sum_{x=1}^\ell \rho_\Psi(x,x)  - \sum_{x=1}^{\ell-1} \rho_\Psi(x,x+1)\,.
\end{equation}
\end{lemma}
\begin{proof}
From the definition of  $\Phi:=\mathbb{V}\Psi$ we have
$$
\|\Phi\|^2 = \sum_{A\in[1,\ell-n+1]^n} | \Phi(A)|^2 = \sum_{X\in  \mathcal{X}_{\ell,n}}  |\Psi(X)|^2 | V^{-1}(V(X))|^{-1} \,,
$$
where $|V^{-1}(V(X))|$ denotes the number of points in the pre-image of $V(X)$.  This number equals one if $X$ is such that $|x_j-x_k|\geq 2$ for all $j\neq k$. Hence 
$$
\|\Phi\|^2 \geq \sum_{\substack{ X\in\mathcal{X}_{\ell,n} \\ |x_j - x_k|\geq 2 \, \forall j\neq k}}  |\Psi(X)|^2  \geq \|\Psi\|^2 - \frac 12 \sum_{x=1}^\ell \langle \Psi | n_x(n_x-1) \Psi\rangle  - \sum_{x=1}^{\ell-1} \langle \Psi | n_x n_{x+1} \Psi \rangle.
$$
Indeed, the norm of $\Psi$ involves a sum over all possible configurations so we need to remove the terms which correspond to $x_i=x_j$ or $x_i=x_j+1$ for some $i \neq j$. The $x_i=x_j$ terms are removed through the term $ \frac 12 \sum_{x=1}^\ell  n_x(n_x-1) $, which is zero if and only if on each site there is at most one particle. Similarly, the terms corresponding to $x_i=x_j+1$ 
are removed through $\sum_{x=1}^{\ell-1}  n_x n_{x+1} $, which is zero if and only if there are no two neighboring sites that are occupied.
With $\|\Psi\|=1$ and the definition of  $\rho_\Psi(x,y)$ this becomes \eqref{eq:normphi_lowerbound}.
\end{proof} 
We shall give a lower bound on the right side of \eqref{eq:normphi_lowerbound} in terms of the energy of $\Psi$. 
\begin{proposition}\label{prop56}
Let $\Psi\in \ell^2_{\rm sym}([1,\ell]^n)$ with $\|\Psi\|=1$. %, and let $\rho(x,y) = \langle \Psi | a^\dagger_x a^\dagger_y a_y a_x \Psi\rangle$ be its  two-particle density.  
Then 
\begin{equation}
\sum_{x=1}^{\ell-1} \rho_\Psi(x+1,x) \leq \frac 4 \ell n(n-1)  + 4 (n-1) \sqrt {\frac n S}  \langle \Psi | H_\ell \Psi\rangle^{1/2}. \label{eq:twobodydensitybound}
\end{equation}
\end{proposition}
\begin{proof}
For $x\neq z$, we have
\begin{align*}
& \rho_\Psi(x,y)\left( 1- \frac{\delta_{z,y}}{2S}\right) - \rho_\Psi(z,y) \left( 1 -\frac{\delta_{x,y}}{2S}\right) 
\\ & = \Re \left\langle \Psi \left| \left( a^\dagger_x\sqrt{1-\frac{n_z}{2S}} - a^\dagger_z \sqrt{1-\frac{n_x}{2S}} \right) n_y  \left( a_x\sqrt{1-\frac{n_z}{2S}} + a_z \sqrt{1-\frac{n_x}{2S}} \right) \right. \Psi \right\rangle \,.
\end{align*}
The Cauchy--Schwarz inequality therefore implies that 
\begin{align*}
& \left|  \rho_\Psi(x,y)\left( 1- \frac{\delta_{z,y}}{2S}\right) - \rho_\Psi(z,y) \left( 1 -\frac{\delta_{x,y}}{2S}\right) \right|^2 
\\ & \leq  \left\langle \Psi \left| \left( a^\dagger_x\sqrt{1-\frac{n_z}{2S}} - a^\dagger_z \sqrt{1-\frac{n_x}{2S}} \right) n_y  \left( a_x\sqrt{1-\frac{n_z}{2S}} - a_z \sqrt{1-\frac{n_x}{2S}} \right) \right.\Psi\right\rangle  \\ &\quad  \times \left\langle \Psi \left| \left( a^\dagger_x\sqrt{1-\frac{n_z}{2S}} + a^\dagger_z \sqrt{1-\frac{n_x}{2S}} \right) n_y  \left( a_x\sqrt{1-\frac{n_z}{2S}} + a_z \sqrt{1-\frac{n_x}{2S}} \right) \right.\Psi \right\rangle \,.
\end{align*}
Moreover,
\begin{align*}
 &\left\langle \Psi \left| \left( a^\dagger_x\sqrt{1-\frac{n_z}{2S}} + a^\dagger_z \sqrt{1-\frac{n_x}{2S}} \right) n_y  \left( a_x\sqrt{1-\frac{n_z}{2S}} + a_z \sqrt{1-\frac{n_x}{2S}} \right) \right.\Psi \right\rangle \\ & \leq 2  \left\langle \Psi \left|  a^\dagger_x \left(1-\frac{n_z}{2S}\right)  n_y   a_x  \right.\Psi \right\rangle + 2 \left\langle \Psi \left|  a^\dagger_z \left(1-\frac{n_x}{2S}\right)   n_y   a_z \right.\Psi \right\rangle \\ & \leq 2 \rho_\Psi(x,y) \left(1-\frac{\delta_{z,y}}{2S}\right) +  2 \rho_\Psi(z,y) \left(1-\frac{\delta_{x,y}}{2S}\right)\,.
\end{align*}
With 
$$
h_x^y := \left( a^\dagger_{x+1}\sqrt{1-\frac{n_x}{2S}} - a^\dagger_x \sqrt{1-\frac{n_{x+1}}{2S}} \right) n_y  \left( a_{x+1}\sqrt{1-\frac{n_x}{2S}} - a_x \sqrt{1-\frac{n_{x+1}}{2S}} \right)
$$
we thus have
\begin{align}\nonumber
& \left|  \rho_\Psi(x+1,y)\left( 1- \frac{\delta_{x,y}}{2S}\right) - \rho_\Psi(x,y) \left( 1 -\frac{\delta_{x+1,y}}{2S}\right) \right|^2 
\\ & \leq  2 \left\langle \Psi \left| h_x^y \right.\Psi\right\rangle    
\left(  \rho_\Psi(x+1,y) \left(1-\frac{\delta_{x,y}}{2S}\right) +  \rho_\Psi(x,y) \left(1-\frac{\delta_{x+1,y}}{2S}\right) \right) \,. \label{eq512}
\end{align}
We note that
$$
S\sum_{x=1}^{\ell-1} \sum_{y=1}^\ell h_x^y = H_\ell \left(\mathcal{N}  -1\right)\,.
$$

For given $y \leq \ell/2$, choose $x_y > y$ such that
$$
\rho_\Psi(x,y) \geq \rho_\Psi(x_y, y) \quad \text {for all $x > y$}\,.
$$
We have 
$$
\rho_\Psi(y+1,y) = \rho_\Psi(x_y,y) + \sum_{w = y+1}^{x_y-1} \left( \rho_\Psi(w,y) - \rho_\Psi(w+1,y) \right)
$$
(where the sum is understood to be zero if $x_y=y+1$). The first term on the right side can be bounded as
$$
\rho_\Psi(x_y,y) \leq \frac 1{\ell-y} \sum_{x=y+1}^\ell \rho_\Psi(x,y) \leq \frac 2 \ell \sum_{x=1}^\ell \rho_\Psi(x,y)
$$ 
using that $y\leq \ell/2$ by assumption. For the second we use the bound \eqref{eq512} above, which implies that 
$$
\left| \rho_\Psi(w,y) - \rho_\Psi(w+1,y) \right| \leq \sqrt 2 \langle \Psi | h_w^y \Psi\rangle^{1/2} \left( \rho_\Psi(w+1,y) + \rho_\Psi(w,y) \right)^{1/2}
$$
for $w\geq y+1$. After summing over $y$ and $w$, using the Cauchy--Schwarz inequality and the fact that $\sum_{x,y}\rho_\Psi(x,y) = n(n-1)$, we thus have the upper bound
$$
\sum_{y\leq \ell/2} \rho_\Psi(y+1,y) \leq \frac {2n(n-1)} \ell  + 2 \sqrt{\frac n S}(n-1) \langle \Psi | H_\ell \Psi\rangle^{1/2}.
$$
If $y>\ell/2$, we use the symmetry of $\rho$ and write 
$$
\rho_\Psi(y+1,y) = \rho_\Psi(y,y+1) = \rho_\Psi(x_y,y+1) + \sum_{w=x_y}^{y-1} \left( \rho_\Psi(w+1 ,y+1) - \rho_\Psi(w, y+1) \right)
$$
instead, where $x_y$ is now defined by minimizing $\rho_\Psi(x,y+1)$ for $x\leq y$. Proceeding as above, 
we finally conclude the desired estimate.
\end{proof}
A similar bound holds for $\sum_x \rho_\Psi(x,x)$. 
\begin{proposition}\label{prop57}
Let $\Psi\in \ell^2_{\rm sym}([1,\ell]^n)$ with $\|\Psi\|=1$. %, and let $\rho(x,y) = \langle \Psi | a^\dagger_x a^\dagger_y a_y a_x \Psi\rangle$ be its  two-particle density.  
Then
\begin{equation}\label{cpo}
\sum_{x=1}^\ell \rho_\Psi(x,x)   \leq   \frac 4 \ell n(n-1)  + (4 + \sqrt{3}) (n-1) \sqrt {\frac n S}  \langle \Psi | H_\ell \Psi\rangle^{1/2}  \,.
\end{equation}
\end{proposition}
\begin{proof}
Since  $ \rho_\Psi(x,x)$ vanishes for $S=1/2$, we can assume $S\geq 1$ henceforth. By \eqref{eq512}, 
\begin{align*}
& \left|  \rho_\Psi(x\pm 1,x)\left( 1- \frac{1}{2S}\right) - \rho_\Psi(x,x)  \right|^2 
\\ & \leq  2 \sum_{y=1}^{\ell -1} \left\langle \Psi \left| h_y^x \right.\Psi\right\rangle    
\left(  \rho_\Psi(x\pm 1,x) \left(1-\frac{1}{2S}\right) +  \rho_\Psi(x,x)  \right) \,.
\end{align*}
It thus follows from the Cauchy--Schwarz inequality that 
\begin{align*}
&\sum_{x=1}^\ell \rho_\Psi(x,x) \leq  2 \left( 1 -\frac 1{2S}\right) \sum_{x=1}^{\ell -1} \rho_\Psi(x+1,x) \\ & \quad +  \sqrt{2(n-1) /S }    \left\langle \Psi | H_\ell  \Psi\right\rangle^{1/2}    
\left( 2 \sum_{x=1}^{\ell -1}  \rho_\Psi(x+ 1,x) \left(1-\frac{1}{2S}\right) + \sum_{x=1}^\ell \rho_\Psi(x,x)  \right)^{1/2}\,.
\end{align*}
In the last line, we can make the rough bounds $2 \sum_{x=1}^{\ell -1}  \rho_\Psi(x+ 1,x) \leq n(n-1)$ and $\sum_{x=1}^\ell \rho_\Psi(x,x)\leq n(n-1)$, and for the term in the first line we use \eqref{eq:twobodydensitybound}. Using also $S\geq 1$, this completes the proof of \eqref{cpo}.
\end{proof}

\smallskip

\noindent {\it Step 7. Final estimate.}  Recall the definition \eqref{def:PE0} of $P_{E_0,n}$. It follows from Prop.~\ref{prop:hamiltonianLaplacianBound} that
$$
P_{E_0,n} H_\ell \geq S P_{E_0,n} \mathbb{V}^\dagger (-\Delta_n^{\ell-n+1}) \mathbb{V}P_{E_0,n}
$$
and from Lemma~\ref{lem:phinormlower}, Prop.~\ref{prop56} and Prop.~\ref{prop57} that
$$
P_{E_0,n} \mathbb{V}^\dagger \mathbb{V} P_{E_0,n} \geq P_{E_0,n} (1 - \delta)  
$$
where
\begin{equation}\label{def:delta}
\delta =    \frac {8N_0^2} \ell   + 9 N_0 \sqrt {\frac {N_0 E_0} S}  
 =  \left( 2 + \frac 9{\sqrt 8} \right)    \frac {E_0^2 \ell^3}{S^2}.     
\end{equation}
Here we used \eqref{def:N_0}. We shall choose the parameters such that $\delta \ll 1$ for large $\beta$. The min-max principle readily implies that the eigenvalues of $H_\ell$ in the range  $P_{E_0,n}$ are bounded from below by the corresponding ones of $S(1-\delta)(-\Delta_n^{\ell-n+1})$. In particular, 
for any $\beta>0$
$$
\Tr P_{E_0,n} e^{-\beta H_\ell} \leq \Tr e^{\beta S(1-\delta) \Delta_n^{\ell-n+1}}\,.
$$

Note that the Laplacian $\Delta_n^{\ell-n+1}$ depends on $n$, besides the particle number, also via the size of the interval $[1,\ell-n+1]$. For a lower bound, we can increase the interval size back to $\ell$, all eigenvalues are clearly decreasing under this transformation. In particular,
\begin{align}\nonumber
 \Tr e^{-\beta H_\ell} \id_{H_\ell < E_0} &  \leq  (2 S\ell  +1) \sum_{n=0}^{\lfloor N_0\rfloor}  \Tr e^{\beta S(1-\delta) \Delta_n^{\ell}} \\
 &  \leq  (2 S\ell  +1) (N_0+1 ) \prod_{m=1}^{\ell-1} \left( 1 - e^{-\beta S (1-\delta) \epsilon(\pi m/\ell) } \right)^{-1} \label{tg2}
 \end{align}
where $\epsilon(p) = 2 (1-\cos p)$ is the dispersion relation of the discrete Laplacian on $[1,\ell]$.

Combining \eqref{tg} and \eqref{tg2}, we have thus shown that
\begin{align*}
f_\ell(\beta,S) & \geq -\frac1{\beta \ell} \ln \left( 1+ (2S\ell+1) (N_0+1 ) \prod_{m=1}^{\ell-1} \left( 1 - e^{-\beta S (1-\delta) \epsilon(\pi m/\ell) } \right)^{-1}\right)
\\ & \geq  \frac1{\beta \ell} \sum_{m=1}^{\ell-1} \ln  \left( 1 - e^{-\beta S (1-\delta) \epsilon(\pi m/\ell) } \right) -\frac1{\beta \ell} \ln \left( 1+ (2S\ell+1) (N_0+1 ) \right)\,,
\end{align*}
with $\delta$ defined in \eqref{def:delta}, $N_0 = E_0 \ell^2/(2S)$ and $E_0 = \OO( \ell \beta^{-3/2} S^{-1/2} ( \ln(\beta S))^{1/2} \ln (\beta S^3))$. 
Since $\epsilon(p)$ is increasing in $p$, we further have
$$
\frac 1{\beta \ell} \sum_{m=1}^{\ell-1} \ln  \left( 1 - e^{-\beta S (1-\delta) \epsilon(\pi m/\ell) } \right) \geq \frac 1{\pi \beta} \int_0^\pi \ln(1-e^{-\beta S (1-\delta) \epsilon(p)}) dp.
$$
The  error terms compared to the desired expression
$$
\frac 1{\pi \beta} \int_0^\pi \ln(1-e^{-\beta S \epsilon(p)}) dp = \OO\left( \beta^{-3/2} S^{-1/2}\right)
$$
are thus 
$$
\ell^5 \frac {\ln(\beta S)}{ (\beta S)^3} \left( \ln (\beta S^3) \right)^2  \quad \text{and} \quad  (\beta S)^{1/2}\ell^{-1} \ln \left( S \ell N_0\right)
$$
which leads to a choice of $\ell = C  (\beta S)^{1/2 + 1/12}  (\ln (\beta S^3))^{-1/3}$ and a relative error of the order $ (\beta S)^{-1/12} \ln(\beta S) (\ln (\beta S^3))^{1/3}$. Note that for this choice the condition $\ell \geq \ell_0/2$ of Lemma~\ref{lem53} is fulfilled exactly when this error is small. 

Finally, we note that (compare with \cite[Eqs. (5.42) and (5.43)]{CGS}) 
$$
 \int_0^\pi \ln(1-e^{-\beta S \epsilon(p)}) dp \geq \frac 1{ (\beta S)^{1/2}} \int_0^\infty \ln(1-e^{-p^2}) dp - \OO( (\beta S)^{-3/2})
$$
for large $\beta S$. 
This completes the proof of the lower bound. \hfill\qed

\appendix
\section{Upper bound in two dimensions}\label{sec:appendix}

In two dimensions we consider the ferromagnetic Heisenberg model with nearest neighbor interactions on the square lattice $\Z^2$. It is defined in terms of the Hamiltonian
\begin{equation}\label{eq:heisenberg_ham_2d}
H_\Lambda :=  \sum_{\langle x,y \rangle \subset \Lambda}(S^2- \spinv_{x} \cdot \spinv_y)\,,
\end{equation}
where $\langle x,y \rangle$ denotes a pair of nearest neighbors and  $\Lambda$ is a finite subset of $\Z^2$. We denote the free energy in the thermodynamic limit by
\begin{equation} \label{eq:free_energy_2d}
f^{\rm{2d}}(\beta,S) : = \lim_{\Lambda \to \Z^2} f^{\rm 2d}_{\Lambda}(\beta,S) = - \lim_{\Lambda \to \Z^2} \frac 1 {\beta |\Lambda|} \Tr e^{-\beta H_\Lambda} \,.
\end{equation}
The limit has to be understood via a suitable sequence of increasing domains, e.g., squares of side length $L$ with $L\to\infty$. 

For $d=2$ we have the following upper bound.
\begin{proposition}\label{thm:2D_upper}
Consider the Hamiltonian \eqref{eq:heisenberg_ham_2d} and the corresponding free energy \eqref{eq:free_energy_2d}. Let 
 \begin{equation}
  C_2:=\frac{1}{(2\pi)^2}\int_{\mathbb{R}^2}\ln \big(1-e^{-p^2}\big)dp = \frac{-\zeta(2)}{4\pi}=-\frac{\pi}{24}. \label{def:C2}
\end{equation}
Then, for any $S\geq 1/2$, we have 
\begin{equation}\label{fe ub asympt2}
f^{\rm{2d}}(\beta,S) \leq C_2 S^{-1} \beta^{-2} \left(1 - \OO( (\beta S)^{-1/3} ( \ln \beta S)^{2/3}) \right)
\end{equation}
as $\beta S \to \infty$. 
\end{proposition}

We note that it remains an open problem to derive a corresponding lower bound, i.e., the analogue of Prop.~\ref{prop:lower} in $d=2$ dimensions.

 The proof of Prop.~\ref{thm:2D_upper} differs from the one-dimensional case discussed in Section \ref{sec:up} only in the evaluation of the error terms in Lemmas \ref{lem:ppex} and \ref{lem:ent}. Let $\Gamma$, $\cP$ and $K$ be defined as in \eqref{def:gamma}, \eqref{def:cP} and \eqref{hamd}, with the obvious modifications to $d=2$, for a square-shaped domain $\Lambda_\ell = [1,\ell]^2$. 
 Then the following holds

\begin{lemma}\label{lem:ppex2D}
In the case $d=2$ we have 
\begin{equation}\label{ppex2D}
 \frac{ \Tr_\F \mathcal{P} e^{-\beta K}\cP}{\Tr_\F e^{-\beta K} } \geq 1- \left( \frac{\pi \ell  \ln (1+2 \ell)}{2\beta S}\right)^2 \,. 
 \end{equation} 
\end{lemma}

\begin{proof}
The bound \eqref{eq:gaussian_bound_energy} remains correct in two dimensions. We thus only need to estimate the (now) double sum over the two-dimensional dual lattice 
\begin{equation*}
 \frac{ \Tr_\F n_x e^{-\beta K}}{\Tr_\F e^{-\beta K}}\leq \sum_{p \in \Lambda_\ell^{*\rm D}}\frac{|\phi_p(x)|^2}{e^{\beta S \epsilon(p)}-1}\leq \frac{4}{(\ell+1)^2}\sum_{m=1}^\ell \sum_{n=1}^\ell \frac{1}{e^{\beta S \tilde{\epsilon} (m, n)}-1}
\end{equation*}
where $\tilde{\epsilon} (m,n)=2(2-\cos (\frac{\pi m}{\ell+1})-\cos (\frac{\pi n}{\ell+1}))$. By proceeding as in the proof of Lemma~\ref{lem:ppex}, we have
\begin{equation}\label{eq:nxexp2d}
 \frac{ \Tr_\F n_x e^{-\beta K}}{\Tr_\F e^{-\beta K}}\leq \frac{1}{\beta S} \sum_{m=1}^\ell \sum_{n=1}^\ell \frac{1}{m^2 + n^2} \leq \frac \pi 2 \frac{ \ln (1+ 2\ell)}{\beta S} \,.
\end{equation}
Looking again at \eqref{eq:gaussian_bound_energy} we see that  the summation over $x\in \Lambda_\ell$ yields a factor $\ell^2$, and hence we arrive at the desired bound \eqref{ppex2D}.
\end{proof}

Next we establish the two-dimensional counterpart of the entropy estimate. We have 

\begin{lemma}\label{lem:ent2D}
In the case $d=2$ we have 
\begin{align*}\label{entbd2}
\frac 1 \beta \Tr \Gamma \ln \Gamma &  \leq - \frac 1 \beta \ln \left( \tr_\F \mathcal{P} e^{-\beta K}\cP \right) -  \frac{\Tr_\F K e^{-\beta K} }{\Tr_\F \mathcal{P} e^{-\beta K}\cP }  \\ & \quad + \frac S 2 \left( \frac \pi 2 \ell (\ell +1) \frac {\ln (1+2\ell)}{(\beta S)^2} \right)^2 \left[ \frac{\pi^3}{48} + \frac {\beta S}{\ell^2} \right] \frac{\Tr_\F e^{-\beta K} }{\Tr_\F \mathcal{P} e^{-\beta K}\cP } 
\,.  
\end{align*}
\end{lemma}

\begin{proof}
As in the case of the previous lemma, the only difference with regard to the one-dimensional case lies in the estimation of the $p$ sums in \eqref{eq:entestimate1D}.  By proceeding similarly as above, we obtain 
$$
\sum_{p \in \Lambda_\ell^{*\rm D}}  \frac {2S \epsilon(p) }{e^{\beta  S \epsilon(p)} -1} \leq \frac {\pi^3}{48} S \frac{(\ell+1)^2}{(\beta S)^{2}}
$$ 
as well as
$$
 \sum_{p \in \Lambda_\ell^{*\rm D}}  \frac { S\epsilon(p)  }{ \left( \sinh \tfrac 12 \beta S \epsilon(p) \right)^2 } \leq S \frac{ (\ell+1)^2}{(\beta S)^2} \sum_{m=1}^\ell \sum_{n=1}^\ell \frac{1}{m^2 + n^2} \leq  \frac \pi 2 S \frac{ (\ell+1)^2}{(\beta S)^2}  \ln(1+2\ell)\,.
$$
In combination with  \eqref{eq:nxexp2d} this yields the desired result.
\end{proof}

It remains to obtain the two-dimensional counterpart of the final estimate of the free energy. The Gibbs variational principle together with Lemma \ref{lem:ppex2D} and Lemma \ref{lem:ent2D}   implies  that for $C (\beta S)^{1/2} \leq \ell  \ll \beta S / \ln (\beta S)$  
\begin{equation*}
\begin{aligned}
f_{\Lambda_\ell}^{2d,\rm D}(\beta,S)
&\leq -\frac{1}{\beta \ell^2}\ln \left( \tr_\F e^{-\beta K} \right)-\frac{1}{\beta \ell^2}\ln \left(1-\frac{C \ell^2 \ln^2 \ell}{(\beta S)^2}\right) + CS\frac{\ell^2 \ln^2 \ell}{(\beta S)^4}
\end{aligned}
\end{equation*}
for a suitable constant $C>0$. 
The first term on the right side  equals
\begin{equation}
 - \frac 1 {\beta \ell^2} \ln \left( \tr_\F e^{-\beta K} \right) = \frac 1 {\beta\ell^2} \sum_{p\in \Lambda_\ell^{*\rm D} } \ln ( 1- e^{-\beta S \epsilon(p)} )\,.
\end{equation} 
By monotonicity, we can again bound the sum in terms of the corresponding integral, i.e., 
\begin{equation}\label{up:ried2} 
 \frac 1 {\beta\ell^2} \sum_{p\in \Lambda_\ell^{*\rm D} } \ln ( 1- e^{-\beta S \epsilon(p)} ) \leq  \frac 1 {\beta \pi^2} \left( 1 + \ell^{-1} \right)^2  \int_{[\frac\pi {\ell+1},\pi]^2} \ln ( 1- e^{-\beta S \epsilon(p)} ) dp \,. 
\end{equation}
The missing term is now bounded by
$$
- \frac 2 {\beta \pi^2}   \int_{[0,\frac\pi {\ell+1}]\times[0,\pi]} \ln ( 1- e^{-\beta S \epsilon(p)} ) dp \leq - \frac 1 {\beta (\beta S)^{1/2} (\ell+1)}   \int_{\R_+} \ln ( 1- e^{-p^2 } ) dp\,.
$$
Furthermore, since $\epsilon(p)\leq |p|^2$ we have
\begin{align}\nonumber
\frac 1 { \pi^2} \int_{[0,\pi]^2} \ln ( 1- e^{-\beta S \epsilon(p)} ) dp & \leq \frac 1 { (2\pi)^2} \int_{\R^2} \ln ( 1- e^{-\beta S |p|^2}  )dp+ \frac C { (\beta S)^\alpha }  \\
& =C_2 (\beta S)^{-1} + \frac C { (\beta S)^\alpha }  \label{4.28}
\end{align}
for $\alpha>0$ arbitrary, some $C>0$ (depending on $\alpha$), and $C_2$ defined in (\ref{def:C2}). For $\ell$ satisfying  $\ell \ln\ell \ll \beta S$ and $\ell \gg (\beta S)^{1/2}$ all the error terms are small compared to the main term. The desired upper bound stated in Prop.~\ref{thm:2D_upper} is obtained by choosing $\ell = C (\beta S)^{5/6} (\ln \beta S)^{-2/3}$. 
\hfill\qed

\bigskip

\noindent \textbf{Acknowledgments.} The work of MN was supported by the National Science Centre (NCN) project Nr. 2016/21/D/ST1/02430. The work of RS was supported by the European Research Council (ERC) under the European Union\rq s Horizon 2020 research and innovation programme (grant agreement No 694227).

\end{document}